\documentclass[a4paper,UKenglish,cleveref, autoref, thm-restate]{lipics-v2021}

\pdfoutput=1 
\hideLIPIcs  

\title{Improved Bound on the Number of Pseudoline Arrangements via the Zone Theorem}

\author{Justin Dallant}{Department of Computer Science, Aarhus University, Denmark} {justindallant@cs.au.dk}{https://orcid.org/0000-0001-5539-9037}{}
\authorrunning{J. Dallant}

\Copyright{Justin Dallant}



\category{} 

\acknowledgements{I thank anonymous reviewers and Günter Rote for their helpful comments on an earlier version of this paper.}

\nolinenumbers
\hideLIPIcs{}

\EventEditors{}
\EventNoEds{2}
\EventLongTitle{}
\EventShortTitle{arXiv 2023}
\EventAcronym{arXiv}
\EventYear{2023}
\EventDate{}
\EventLocation{}
\EventLogo{}
\SeriesVolume{}
\ArticleNo{}

\usepackage{numprint}

\usepackage{mathtools}
\newcommand{\DeclareAutoPairedDelimiter}[3]{%
  \expandafter\DeclarePairedDelimiter\csname Auto\string#1\endcsname{#2}{#3}%
  \begingroup\edef\x{\endgroup
    \noexpand\DeclareRobustCommand{\noexpand#1}{%
      \expandafter\noexpand\csname Auto\string#1\endcsname*}}%
  \x}

\DeclareAutoPairedDelimiter\ceil{\lceil}{\rceil}
\DeclareAutoPairedDelimiter\floor{\lfloor}{\rfloor}

\renewcommand{\epsilon}{\varepsilon}

\DeclareMathOperator{\h}{h_2}

\usepackage{thmtools}
\usepackage{thm-restate}

\usepackage{tikz}
\usepackage{algpseudocode}

\begin{document}

\maketitle

\begin{abstract}
Pseudoline arrangements are fundamental objects in discrete and computational geometry, and different works have tackled the problem of improving the known bounds on the number of simple arrangements of $n$ pseudolines over the past decades. The lower bound in particular has seen two successive improvements in recent years (Dumitrescu and Mandal in 2020 and Cortés Kühnast et al.\ in 2024). Here we focus on the upper bound, and show that for large enough $n$, there are at most $2^{0.6496n^2}$ different simple arrangements of $n$ pseudolines. This follows a series of incremental improvements starting with work by Knuth in 1992 showing a bound of roughly $2^{0.7925n^2},$ then a bound of $2^{0.6975n^2}$ by Felsner in 1997, and finally the previous best known bound of $2^{0.6572n^2}$ by Felsner and Valtr in 2011. The improved bound presented here follows from a simple argument to combine the approach of this latter work with the use of the Zone Theorem.
\end{abstract}

\section{Introduction}

An arrangement of pseudolines in the Euclidean plane is a finite set of $n$ simple curves extending to infinity in both directions such that every two curves intersect at exactly one point where they cross. They have been the subject of numerous works, the earliest of which going back to at least the 1920's~\cite{levi1926teilung}. We refer the reader to the corresponding chapter of the Handbook \cite[Chapt.\ 5]{handbook} for a thorough overview of the subject.

We will always assume pseudoline arrangements are \emph{simple}, meaning that no three pseudolines intersect at a common point. Following \cite{felsner2011coding}, we will consider our arrangements to be \emph{marked}\footnote{Note that the final bound we obtain here applies similarly to unmarked arrangements or arrangements in the projective plane, as these differences only affect lower order factors.}, meaning that they come with a distinguished unbounded cell called the \emph{north-cell} (the unique cell lying on the other side of all pseudolines is called the \emph{south-cell}). Two pseudoline arrangements are \emph{isomorphic} if one can be mapped to the other one by an orientation preserving homeomorphism of the plane which preserves the north-cell. 

\begin{figure}[t]
    \centering
    \includegraphics[width=0.9\linewidth]{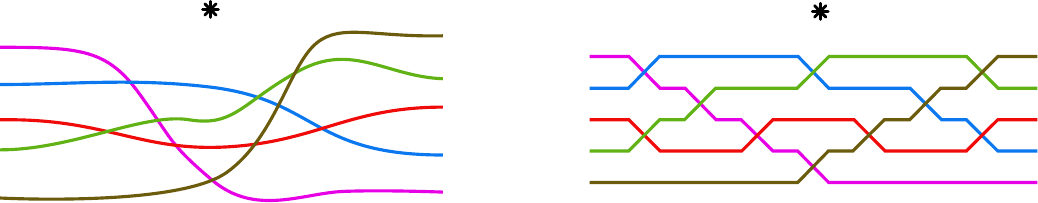}
    \caption{An arbitrarily drawn pseudoline arrangement (left) and the same arrangement drawn as a wiring diagram (right). In both cases, the north-cell is marked by a star.}
    \label{fig:wiring_diagram}
\end{figure}
\subparagraph*{Wiring diagrams.}
One particularly nice way of drawing (in our case simple) pseudoline arrangements, introduced by Goodman~\cite{goodman1980proof} is through \emph{wiring diagrams}. In a wiring diagram, the $n$ pseudolines are constrained to lie on $n$ horizontal lines (or ``wires''), except around points where they cross another pseudoline and move to a neighboring wire. Figure~\ref{fig:wiring_diagram} shows an example of an arbitrarily drawn pseudoline arrangement next to a realization of the same arrangement as a wiring diagram.

\subparagraph*{Number of pseudoline arrangements.}

 It is known that the number $B_n$ of non-isomorphic arrangements of $n$ pseudolines is $2^{\Theta(n^2)}$ but the leading hidden constant in the exponent is not known. Let $c^- = \liminf_{n\to \infty}{\frac{\lg B_n}{n^2}}$ and $c^+ = \limsup_{n\to \infty}{\frac{\lg B_n}{n^2}}$ (throughout the paper, we use $\lg$ to denote the logarithm in base $2$). It is in fact an open question whether $c^- = c^+$, i.e.\ whether the limit of $\frac{\lg B_n}{n^2}$ exists. The first bound on $c^-$ was given by Goodman and Pollak in the 1980's~\cite{goodman1980combinatorial} who showed $c^- \geq 1/8$. This was improved in the 1990's to $c^- \geq 1/6$ by Knuth~\cite{knuth1992axioms}, who also showed $c^+ < 0.79249$ and conjectured $c^+ \leq 0.5$. The upper bound was then improved by Felsner~\cite{felsner1996number} to $c^+ < 0.69743$ and finally to $c^+ < 0.69743$ in 2011 by Felsner and Valtr~\cite{felsner2011coding} who also showed $c^- \geq 0.1887$. The lower bound saw some recent additional improvements. In 2020, Dumitrescu and Mandal~\cite{dumitrescu2020new} showed $c^- > 0.2083$. In 2024, Cortés Kühnast, Felsner and Scheucher~\cite{kuhnast2024improved} and Dallant~\cite{dallant2024improved} independently discovered similar methods to improve the previous construction, which resulted in a merged paper~\cite{DBLP:conf/compgeom/KuhnastDFS24} showing the currently best known bound of $c^- > 0.2721$.

\begin{figure}[h]
    \centering
    \includegraphics[width=0.5\linewidth]{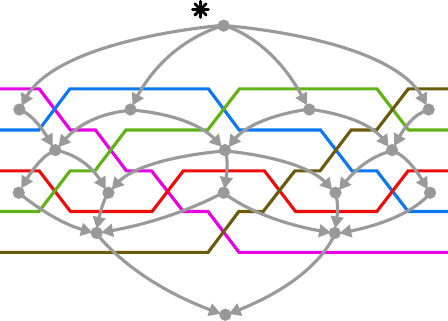}
    \caption{A pseudoline arrangement and its associated directed acyclic dual graph.}
    \label{fig:dual_graph}
\end{figure}
\subparagraph*{Cutpaths and upper bound.} Consider the directed acyclic graph which is dual to the wiring diagram, with edges oriented from north to south (see Figure \ref{fig:dual_graph}). A path from the north-cell to the south-cell in the graph is called a \emph{cutpath} of the arrangement. Informally, the set of all cutpaths represents all combinatorially distinct ways one can insert a new pseudoline in the pseudoline arrangement. The approach of Felsner and Valtr for the upper bound, as did that of Knuth earlier, relies on bounding the quantity $\gamma_n$, defined as the maximum number of cutpaths an arrangement of $n$ pseudolines can have. As shared in a remark by Knuth~\cite[p.\ 97]{knuth1992axioms}, an argument by Bern shows that one can use the sharpest version of the Zone Theorem to show that $\gamma_n \leq O(2.711^n)$, yielding the bound $c^+ < 0.7194$. Knuth also conjectured that the ``bubblesort arrangements'' of size $n$, which (in a slightly different setting than the one presented here) have $n2^{n-2}$ cutpaths, maximize this number.\footnote{Knuth verified this to be true for $n\leq 9$. However, in the setting of the present paper, this is not true even for small $n$, where the ``odd-even sort arrangements'' (see Section \ref{section:many_cutpaths}) maximize this number at least up to $n=10$ (note that in Knuth's setting, these also have $n2^{n-2}$ cutpaths). Thanks to Günter Rote for verifying this, and for pointing out that the arrangement of Figure \ref{fig:bubble20} was mistakenly referred to as a ``bubblesort arrangement'' in an earlier version of this paper (the same confusion between bubblesort arrangements and odd-even sort arrangements also appears in \cite{felsner2011coding}).} If true, this would yield $c^+ \leq 0.5$. Felsner and Valtr show $\gamma_n \leq O(2.487^n)$ without using the Zone Theorem, resulting in the previously best known bound of $c^+ < 0.6572$. They also report that a construction by Ond\v{r}ej Bílka shows $\gamma_n\geq \Omega(2.076^n)$, thus disproving Knuth's conjecture on the maximum number of cutpaths.

We use a simple argument to combine their method with the use of the Zone Theorem, thus improving the upper bound on the number of cutpaths to roughly $\gamma_n \leq O(2.461^n)$, in turn resulting in $c^+ < 0.6496$. The argument, given in the next section, results in a bound expressed as the maximum of a certain function over some domain. We then explain how to obtain an upper bound on this function through a branch and bound through linear relaxation based approach, resulting in our final bound. We also describe in Section \ref{section:many_cutpaths} a family of arrangements showing $\gamma_n \geq \Omega(2.083^n)$, a slight improvement on the construction by Bílka. 

\section{Improved upper bound on the number of pseudoline arrangements}

We start by giving a brief overview of the approach of Felsner and Valtr~\cite{felsner2011coding} to bound the maximum possible number $\gamma_n$ of cutpaths in an arrangement of $n$ pseudolines. Using the fact that $B_{n+1} \leq \gamma_n \cdot B_n$, this then yields a bound on $B_n$ by induction.

Consider an arrangement $A$ of $n$ pseudolines represented as a wiring diagram, together with a cutpath $p$ traversing $A$ from the north cell to the south cell. The path of $p$ can be described by choosing for each cell it reaches, which pseudoline bounding the cell from below it crosses to exit the cell. We classify these exits into different categories. Consider a cell $c$ with $d$ exits (i.e.\ $d$ pseudolines bounding it from below), which we label $1,2,\ldots,d$ from left to right. If $d=1$, we call the single existing exit the \emph{unique exit} of $c$. Otherwise, we call exit $1$ the \emph{left exit} of $c$ and exit $d$ the \emph{right exit} of $c$. All other exits are called \emph{middle exits} of $c$.

The following fact was proven by Knuth.
\begin{lemma}[{\cite{knuth1992axioms}}]\label{lemma:middles}
    Any pseudoline in $A$ can appear at most once as a middle exit of a cell visited by $p$.
\end{lemma}

If the cutpath crosses $k$ pseudolines as a middle exit and $u$ as a unique exit, one can thus ecode it by choosing the $k$ pseudolines the cutpath will cross as a middle exit, and fixing some binary string $\beta_p$ of length $n-k-u$. Any time a cell is reached, we check if it has a unique exit or if one of the middle exits is among the $k$ chosen, and we take that exit if this is the case. If it is not the case (which happens at most $n-k-u$ times), we consume one entry of $\beta_p$ to choose whether we take the left or right exit. This encoding shows that for any fixed $k$, the number of cutpaths taking $k$ middle exits and $u$ unique exits is at most $\binom{n}{k}2^{n-k-u}$.

This bound can be further refined by analyzing the repartition of middle exits among visited cells more carefully, showing that the number of cutpaths taking $k$ middle exits and $u$ unique exits is at most 
\[\binom{n-u}{k}\left(\frac{n}{n-u}\right)^k 2^{n-k-u}.\]

Here we have encoded the cutpath by describing its course from the north-cell to the south-cell, but one can also choose to encode it in the other direction, by rotating the whole plane by $180$ degrees and describing its course from the south-cell to the north-cell. The crucial observation made by Felsner and Valtr is that any middle exit in the original description corresponds to a unique exit in the rotated description (and similarly any middle exit in the rotated description corresponds to a unique exit in the original description). Felsner and Valtr exploit the refined bound, together with the freedom to choose in which direction to encode a cutpath, to show their best bound of roughly $\gamma_n \leq 4n\cdot 2.487^n$.

\paragraph*{Improving the upper bound.} In order to improve the upper bound, we will pay closer attention to the total number of middle exits among all cells visited by $p$. Let $\Gamma(m,k,u,m',k',u')$ denote the set of cutpaths of $A$ such that when viewed from the north-cell to the south-cell
\begin{itemize}
    \item the total number of middle exits taken is $k$,
    \item the total number of unique exits taken is $u$,
    \item the total number of middle exits among cells visited (including those which were not taken) is $m$,
\end{itemize}
and when viewed from the south-cell to the north-cell
\begin{itemize}
    \item the total number of middle exits taken is $k'$,
    \item the total number of unique exits taken is $u'$,
    \item the total number of middle exits among cells visited (including those which were not taken) is $m'$,
\end{itemize}

Stated more precisely, the bound shown by Felsner and Valtr is the following.
\begin{lemma}[{\cite{felsner2011coding}}]
    For all $m,k,u,m',k',u'$ such that $k \leq n-u \leq m$, we have
    \[|\Gamma(m,k,u,m',k',u')| \leq \binom{n-u}{k}\left(\frac{m}{n-u}\right)^k 2^{n-k-u}.\]
\end{lemma}

The condition $n-u \leq m$ is actually unnecessary here. Indeed, the encoding mentioned in the beginning of the section immediately implies a bound of $\binom{m}{k}2^{n-k-u}$, and when $n-u \geq m$ we have that
\begin{align*}
    \binom{m}{k}/\binom{n-u}{k} = \prod_{i=0}^{k-1} \frac{m-i}{n-u-i}
    \leq \prod_{i=0}^{k-1}\frac{m}{n-u} = \left(\frac{m}{n-u}\right)^k.
\end{align*}
Thus, if $n-u \geq m$, then $\binom{m}{k}2^{n-k-u} \leq \binom{n-u}{k}\left(\frac{m}{n-u}\right)^k 2^{n-k-u}$ and we can slightly extend the result to the following\footnote{Alternatively, the proof by Felsner and Valtr can be easily tweaked so that the condition $n-u \leq m$ becomes unnecessary, by replacing the partitions of $[n]$ into $n-u$ subsets they consider in their proof with colorings of $[n]$ using $n-u$ colors.}.

\begin{lemma}\label{lemma:cutpath_bound}
    For all $m,k,u,m',k',u'$ such that $k \leq m$ and $k \leq n-u$, we have
    \[|\Gamma(m,k,u,m',k',u')| \leq \binom{n-u}{k}\left(\frac{m}{n-u}\right)^k 2^{n-k-u}.\]
    By symmetry we also have that for $k' \leq m'$ and $k' \leq n-u'$,
    \[|\Gamma(m,k,u,m',k',u')| \leq \binom{n-u'}{k'}\left(\frac{m'}{n-u'}\right)^{k'} 2^{n-k'-u'}.\]
\end{lemma}

Note that the conditions $k\leq m$ and $k \leq n-u$ (and similarly for $k'\leq m'$ and $k' \leq n-u'$) always hold, as a cutpath cannot take more middle exits than there are such exits, and a taken exit cannot be both a middle and a unique exit simultaneously. The correspondence between middle exits and unique exits in the rotated arrangement implies $k\leq u'$ and $k'\leq u$. Lemma \ref{lemma:middles} implies $m\leq n$ and $m'\leq n$. 

In order to get a tighter grip on $m$ and $m'$, our simple idea is to complement the use of Lemma \ref{lemma:middles} to bound them separately with the use of (the tight version of) the Zone Theorem to bound them conjointly (one could obtain the same bound we obtain here without using Lemma \ref{lemma:middles} at all, but its use makes some of the arguments simpler). Call \emph{zone} of a pseudoline the set of cells supported by that pseudoline, and call \emph{complexity of the zone} the sum of the number of edges over all the cells in the zone (an edge may be counted twice if it appears in two different cells of the zone). Then the following holds.
\begin{theorem}[Zone Theorem {\cite{bern1990horizon,pinchasi2011zone}}]
    In an arrangement of $n+1$ pseudolines, the complexity of the zone of any given pseudoline is at most $9.5n-3$.
\end{theorem}

In our setting (after discarding the $2(n+1)$ edges on the pseudoline itself as well as accounting for the $2n$ edges counted twice because they were ``cut'' by the pseudoline) this translates to the following.
\begin{lemma}\label{lemma:zone}
    The sum of the number of edges bounding a cell over all cells visited by a cutpath is at most $5.5n-5$.
\end{lemma}

A similar result (with an additional term due to a slightly different setting), is mentioned in Knuth's book~\cite[p.\ 97]{knuth1992axioms}.

Let us count this number of edges in terms of $m, m',u$ and $u'$. We start by counting only edges bounding cells from below. Every middle or unique exit contributes $1$ to this count. For all $n-u$ visited cells which have more than one exit, we have to add two edges corresponding to the left and right exits, totaling in $m+u+2(n-u) = 2n+m-u$ edges below cells visited by the cutpath. Repeating the same count for edges above the cells by a $180$ degree rotation of the plane gives $2n+m'-u'$, resulting in a total of $4n+m+m'-u-u'$. Lemma \ref{lemma:zone} thus implies $4n+m+m'-u-u' \leq 5.5n$, i.e.\ $m+m'-u-u' \leq 1.5n$.

We can now bound the total number of cutpaths as follows.
\begin{lemma}\label{lemma:rough_bound}
    For any $n$, we have $\gamma_n \leq n^6\max\{F_n(m,k,u,m',k',u')\}_{(m,k,u,m',k',u')\in U_n}$, where 
    \[F_n(m,k,u,m',k',u') := \min\left\{\binom{n-u}{k}\left(\frac{m}{n-u}\right)^k 2^{n-k-u}, \binom{n-u'}{k'}\left(\frac{m'}{n-u'}\right)^{k'} 2^{n-k'-u'}\right\}\]
    and
    \begin{align*}
        U_n := \{(m,k,u,m',k',u') \in \mathbb{N}^6 \mid &k+u \leq n, k \leq m, m\leq n, m'\leq n,
        k'+u' \leq n, k' \leq m',\\
        &k\leq u',\ k' \leq u,\\
        &m+m'-u-u' \leq 1.5n\}.
    \end{align*}
\end{lemma}
\begin{proof}
    By our previous discussions, for all $(m,k,u,m',k',u') \not\in U_n$, $\Gamma(m,k,u,m',k',u') = \emptyset$. It is also immediate that $|U_n| \leq n^6$.
    By Lemma \ref{lemma:cutpath_bound} we have 
    \begin{align*}
        \gamma_n &\leq \sum_{(m,k,u,m',k',u')\in U_n}|\Gamma(m,k,u,m',k',u')| \\
        &\leq \sum_{(m,k,u,m',k',u')\in U_n}F_n(m,k,u,m',k',u')
         \\
        &\leq n^6\max\{F_n(m,k,u,m',k',u') \}_{(m,k,u,m',k',u')\in U_n}. \qedhere
    \end{align*}
\end{proof}

\paragraph*{Bounding $F_n$.}
To obtain a concrete upper bound on $\gamma_n$ using Lemma \ref{lemma:rough_bound}, it remains to bound $F_n(m,k,u,m',k',u')$ over $U_n$. We will do this by bounding $f_n(m,k,u,m',k',u') := \frac{1}{n}\lg(F_n(m,k,u,m',k',u'))$ over $U_n$. 

By heuristic optimization methods, we found that the maximum value seems to be around $1.299$, which one can obtain by setting $m=m'=0.96n$ and $u=k=u'=k'=0.21n$ and taking the limit as $n$ goes to infinity.

We will show that this is indeed close to the truth and that over this domain we can bound the function by $f_n(m,k,u,m',k',u') < 1.2992$.
For any maximization problem $R$, let $\textrm{MAX}(R)$ denote the optimal value for the problem (if the problem has no feasible solution, then let $\textrm{MAX}(R) = -\infty$).
Let $\h$ be the binary entropy function defined over $0\leq x \leq 1$ by \[\h(x)=-x\lg x-(1-x)\lg(1-x),\] with the convention that $0\cdot \lg 0 = 0$.
The following classical bound (see e.g.~\cite[p.~353]{thomas2006elements}) will be useful.
\begin{lemma}
    For all $0\leq q \leq r$, we have $\binom{r}{q} \leq 2^{r\cdot \h(q/r)}$.
\end{lemma}

Using this bound, and setting $\mu = m/n$, $\mu' = m'/n$, $\kappa = k/n$, $\kappa' = k'/n$, $\upsilon = u/n$ and $\upsilon' = u'/n$ we have that the maximum of $f_n(m,k,u,m',k',u')$ over $U_n$ is upper bounded by the optimum value of the following maximization problem\footnote{Note that if we ignore the very last constraint in $P$ (stemming from the Zone Theorem), the maximum of the resulting problem would yield a value corresponding to the bound of Felsner and Valtr.}, denoted by $P$.
\begin{align*}
    P:\quad\quad &\\
    \max \quad &X\\
    \textrm{s.t.} \quad & X \leq (1-\upsilon)\h\left(\frac{\kappa}{1-\upsilon}\right) + \kappa\lg\left(\frac{\mu}{1-\upsilon}\right) + 1-\kappa -\upsilon, \\
    &X \leq (1-\upsilon')\h\left(\frac{\kappa'}{1-\upsilon'}\right) + \kappa'\lg\left(\frac{\mu'}{1-\upsilon'}\right) + 1-\kappa' -\upsilon', \\
    & 0 \leq \mu \leq 1, \hspace{19.7mm} 0 \leq \mu' \leq 1, \\
    & 0 \leq \kappa \leq \mu, \hspace{19.5mm} 0 \leq \kappa' \leq \mu', \\
    & 0 \leq \kappa+\upsilon \leq 1, \hspace{13.5mm} 0 \leq \kappa'+\upsilon' \leq 1, \\
    & \kappa \leq \upsilon',\hspace{25mm} \kappa' \leq \upsilon\\
    & \mu+\mu'-\upsilon-\upsilon' \leq 1.5.
\end{align*}

Due to the concavity of $\h$, we have that  for all $0\leq x\leq 1$ and all $0<c<1$, $\h(x) \leq \h'(c)(x-c) +\h(c)$. We can thus linearize any term of the form $a\h(\frac{b}{a})$ by replacing it with a new variable $y$ and adding new constraints of the form $y \leq \h'(c)b +(\h(c)-c\h'(c))a$ for different choices of constants $0<c<1$. This can only increase the optimum value. Here we choose the values $c= 0.01, 0.02, \ldots 0.99$.

It remains to deal with the terms $\lg\left(\frac{\mu}{1-\upsilon}\right)$ and $\lg\left(\frac{\mu'}{1-\upsilon'}\right)$.
To do so, we rewrite the problem in an equivalent manner by introducing new variables $\ell$ and $\ell'$, replacing the two previous terms by $\lg \ell$ and $\lg \ell'$ respectively, then adding the constraints $\ell(1-\upsilon) = \mu$ and $\ell'(1-\upsilon') = \mu'$.

Note that if $\upsilon \geq \frac{3}{4}$, the first constraint of $P$ together with the constraint $\kappa+\upsilon \leq 1$ and the fact that for all $x > 0$,  $x\lg x \geq \frac{-1}{e\ln 2} \geq -0.54$ and for all $0\leq x\leq 1$, $h_2(x)\leq 1$ ensure that
\begin{align*}
    X &\leq \frac{1}{4} + \kappa(\lg \mu-\lg(1-\upsilon))+1-\kappa-\frac{3}{4} \\
    &\leq \frac{1}{2} - (1-\upsilon)\lg(1-\upsilon) 
    \leq 1.04.
\end{align*}
Thus, for the purpose of showing that $ \textrm{MAX}(P) \leq 1.2992$ we can safely assume $\frac{\mu}{1-\upsilon} \leq 4$ (which is implied by $\upsilon \leq \frac{3}{4}$ and $\mu \leq 1$), or in other words $\ell \leq 4$. Symmetrically, we can also assume $\ell' \leq 4$. For reasons made evident in a moment, we actually parametrize the problem by bounds we impose on $\ell$ and $\ell'$, resulting in the following optimization problem.
\begin{align*}
    P'(\ell_-,\ell_+,\ell'_-,\ell'_+)&:\\
    \max \quad &X\\
    \textrm{s.t.} \quad & X \leq y + \kappa\lg \ell + 1-\kappa -\upsilon, \\
    &X \leq y' + \kappa'\lg \ell' + 1 - \kappa' -\upsilon', \\
    & \forall c \in \{0.01,0.02,\ldots, 0.99\}: y \leq \h(c)\kappa + (\h(c)-c\h'(c))(1-\upsilon), \\
    & \forall c \in \{0.01,0.02,\ldots, 0.99\}: y' \leq \h(c)\kappa' + (\h(c)-c\h'(c))(1-\upsilon'), \\
    & 0 \leq \mu \leq 1, \hspace{19.7mm} 0 \leq \mu' \leq 1, \\
    & 0 \leq \kappa \leq \mu, \hspace{19.5mm} 0 \leq \kappa' \leq \mu', \\
    & 0 \leq \kappa+\upsilon \leq 1, \hspace{13.5mm} 0 \leq \kappa'+\upsilon' \leq 1, \\
    & \kappa \leq \upsilon',\hspace{25mm} \kappa' \leq \upsilon\\
    & \mu+\mu'-\upsilon-\upsilon' \leq 1.5,\\
    &\ell(1-\upsilon) = \mu, \hspace{15.8mm} \ell'(1-\upsilon') = \mu',\\
    &\ell_- \leq \ell \leq \ell_+, \hspace{16.3mm} \ell'_- \leq \ell' \leq \ell'_+.
\end{align*}

From our previous discussion, we know that $\textrm{MAX}(P'(0,4,0,4)) < 1.2992$ implies that $\textrm{MAX}(P) < 1.2992$. We will use a standard branch and bound method based on branching around values of $\ell$ and $\ell'$. In order to do so we introduce one last optimization problem $Q(\ell_-,\ell_+,\ell'_-,\ell'_+)$ obtained from $P'(\ell_-,\ell_+,\ell'_-,\ell'_+)$ by replacing $\lg \ell $ with the constant $\lg \ell_+ $ as well as replacing the constraint $\ell(1-\upsilon) = \mu$ with the two constraints $\ell_-(1-\upsilon) \leq \mu$ and $\ell_+(1-\upsilon) \geq \mu$ (and similar modifications for $\ell'$). We have that $\textrm{MAX}(P'(\ell_-,\ell_+,\ell'_-,\ell'_+)) \leq \textrm{MAX}(Q(\ell_-,\ell_+,\ell'_-,\ell'_+))$. Moreover, $Q(\ell_-,\ell_+,\ell'_-,\ell'_+)$ is a linear program.

\newpage
Consider the following procedure. 
\begin{center}
    \begin{minipage}{0.75\textwidth}
        \centering{$\textrm{Bound}(\ell_-,\ell_+,\ell'_-,\ell'_+):$}
        \begin{itemize}
        \item Solve $Q(\ell_-,\ell_+,\ell'_-,\ell'_+)$ and let $\Tilde{Q} = \textrm{MAX}(Q(\ell_-,\ell_+,\ell'_-,\ell'_+))$ be the optimum returned value.
        \item If $\Tilde{Q} \geq 1.2992$, let $\ell_\bullet = (\ell_- + \ell_+)/2$, $\ell'_\bullet = (\ell'_- + \ell'_+)/2$ and make the four following recursive calls:
        \begin{itemize}
            \item $\textrm{Bound}(\ell_-,\ell_\bullet,\ell'_-,\ell'_\bullet)$,
            \item $\textrm{Bound}(\ell_-,\ell_\bullet,\ell'_\bullet,\ell'_+)$,
            \item $\textrm{Bound}(\ell_\bullet,\ell_+,\ell'_-,\ell'_\bullet)$,
            \item $\textrm{Bound}(\ell_\bullet,\ell_+,\ell'_\bullet,\ell'_+)$.
        \end{itemize}
        \item Return ``Success''.
    \end{itemize}
    \end{minipage}
\end{center}

\begin{lemma}
    The successful termination of $\textrm{Bound}(0,4,0,4)$ implies that $\textrm{MAX}(P) < 1.2992$.
\end{lemma}
\begin{proof}
    We prove by induction on the recursion depth that the successful termination of $\textrm{Bound}(\ell_-,\ell_+,\ell'_-,\ell'_+)$ implies $\textrm{MAX}(P'(\ell_-,\ell_+,\ell'_-,\ell'_+)) <1.2992$. 
    
    If the recursion depth is $0$ (i.e.\ no recursive calls were made) then we immediately have $\textrm{MAX}(P'(\ell_-,\ell_+,\ell'_-,\ell'_+)) \leq \textrm{MAX}(Q(\ell_-,\ell_+,\ell'_-,\ell'_+)) < 1.2992$.
    Otherwise, assume the recursion depth is $d>0$. We have
    \begin{align*}
        \textrm{MAX}(P'(\ell_-,\ell_+,\ell'_-,\ell'_+)) &= 
        \max\{\textrm{MAX}(P'(\ell_-,\ell'_\bullet,\ell'_-,\ell'_\bullet)),
        \textrm{MAX}(P'(\ell_-,\ell'_\bullet,\ell'_\bullet,\ell'_+)),\\
        & \phantom{=\max\{}\textrm{MAX}(P'(\ell'_\bullet,\ell_+,\ell'_-,\ell'_\bullet)),
        \textrm{MAX}(P'(\ell'_\bullet,\ell_+,\ell'_\bullet,\ell'_+))\}.
    \end{align*}
    
    By the inductive hypothesis, the successful termination of the recursive calls (which all have recursion depth at most $d-1$) shows that each term on the r.h.s.\ is less than $1.2992$.
   
    It follows that 
    $
        \textrm{MAX}(P'(\ell_-,\ell_+,\ell'_-,\ell'_+)) < 1.2992.
    $
    
    As the successful termination of $\textrm{Bound}(\ell_-,\ell_+,\ell'_-,\ell'_+)$ implies finite recursion depth, the claim follows by induction.
    In particular, the successful termination of $\textrm{Bound}(0,4,0,4)$ implies $\textrm{MAX}(P'(0,4,0,4)) < 1.2992$, which in turn implies $\textrm{MAX}(P) < 1.2992$.
\end{proof}

We have checked that $\textrm{Bound}(0,4,0,4)$ does indeed successfully terminate. Note that to account for floating point errors we have replaced the value $1.2992$ with $1.29915$ in the actual implementation (this leads to a recursion depth of $15$ and roughly $30\, 000$ recursive calls in total). By Lemma \ref{lemma:rough_bound} we then have that $\gamma_n \leq n^6 \cdot 2^{\alpha\cdot n}$ where $\alpha< 1.2992$ (as mentioned earlier, this is close to optimal, which can be seen by taking $\mu = \mu' = 0.96$ and $\kappa=\kappa'=\nu=\nu' = 0.21$). This together with $B_{n+1} \leq \gamma_n B_n$ yields our final bound.
\begin{theorem}
    For large enough $n$, the number $B_n$ of simple arrangements of $n$ pseudolines is at most $2^{0.6496n^2}$.
\end{theorem}
\begin{proof}
    Let $n_0$ be an integer such that for all $n\geq n_0$, $\gamma_n \leq 2^{\alpha\cdot n}$, where $\alpha< 1.2992$. 
    For any $n\geq n_0$, we thus have  $B_{n+1} \leq 2^{\alpha\cdot n} B_n$. By induction, it follows that \[B_{n+1} \leq 2^{\alpha\cdot (n_0+(n_0+1)+\ldots n)} B_{n_0} \leq 2^{\alpha\frac{(n-n_0+1)(n_0+n)}{2}}B_{n_0}.\]
    As $\alpha< 1.2992$ this last expression is at most $2^{0.6496n^2}$ for large enough $n$.
\end{proof}

\newpage
\section{Construction of pseudoline arrangements with many cutpaths}\label{section:many_cutpaths}

\begin{figure}[h]
    \centering
    \includegraphics[width=0.75\linewidth]{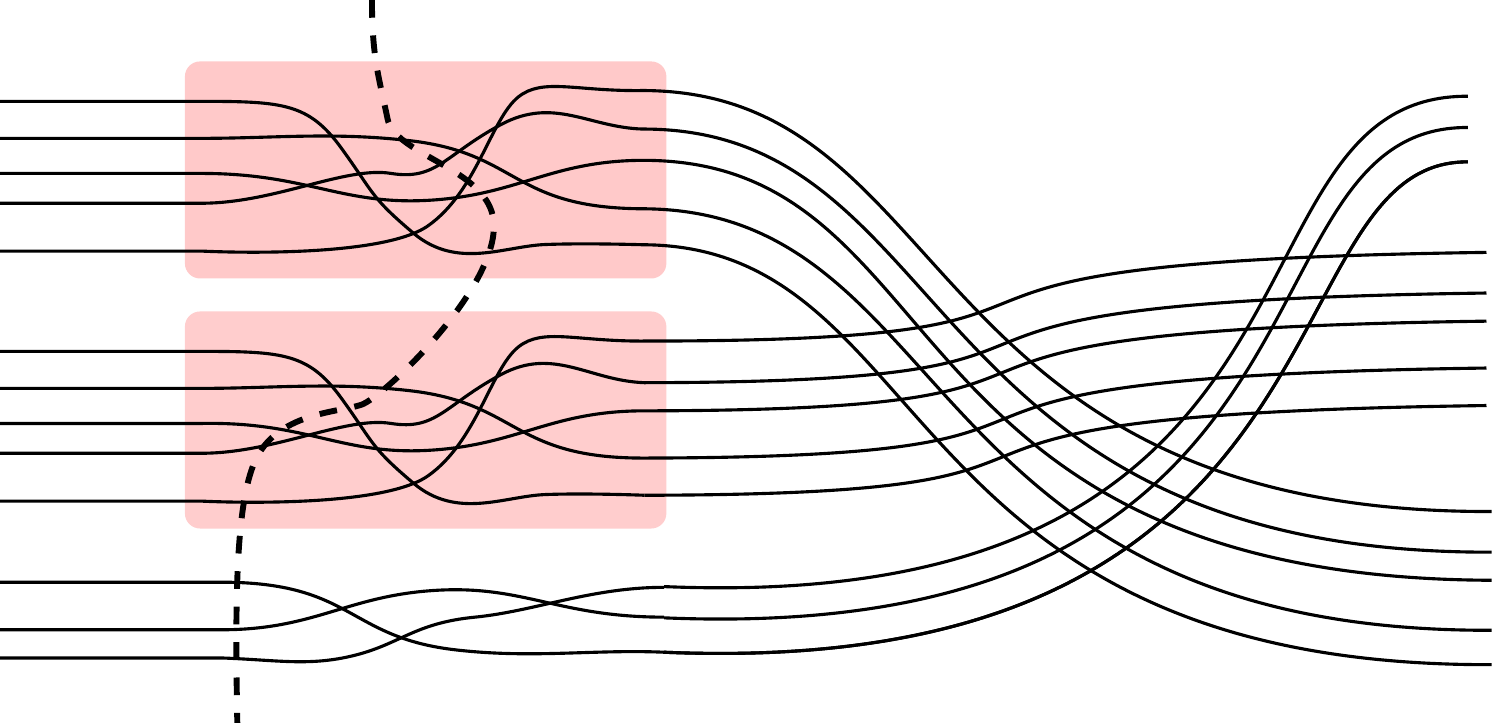}
    \caption{Arrangement of $13$ pseudolines obtained by stacking two copies of the same arrangement of $5$ pseudolines together with $3$ additional pseudolines. The dotted line represents a cutpath.}
    \label{fig:stack}
\end{figure}
Here we briefly describe a way to construct, for any $n>0$, an arrangement of $n$ pseudolines having at least $\Omega(2.083^n)$ cutpaths. Note that it is enough for this purpose to construct a single arrangement of $n_0$ pseudolines with at least $2.083^{n_0}$ cutpaths. This is because for any $n>n_0$, we can stack $\floor{n/n_0}$ disjoint copies of this arrangement on top of each other, add an additional arbitrary arrangement on $n - n_0 \floor{n/n_0}$ pseudolines and then insert the missing crossings between pseudolines belonging to different copies in an arbitrary manner (as illustrated in Figure~\ref{fig:stack}). This results in an arrangement of $n$ pseudolines with at least $(2.083^{\floor{n/n_0}})^{n_0} \geq 2.083^{n-n_0}$ cutpaths.

\begin{figure}[h]
    \centering
    \includegraphics[width=0.8\linewidth]{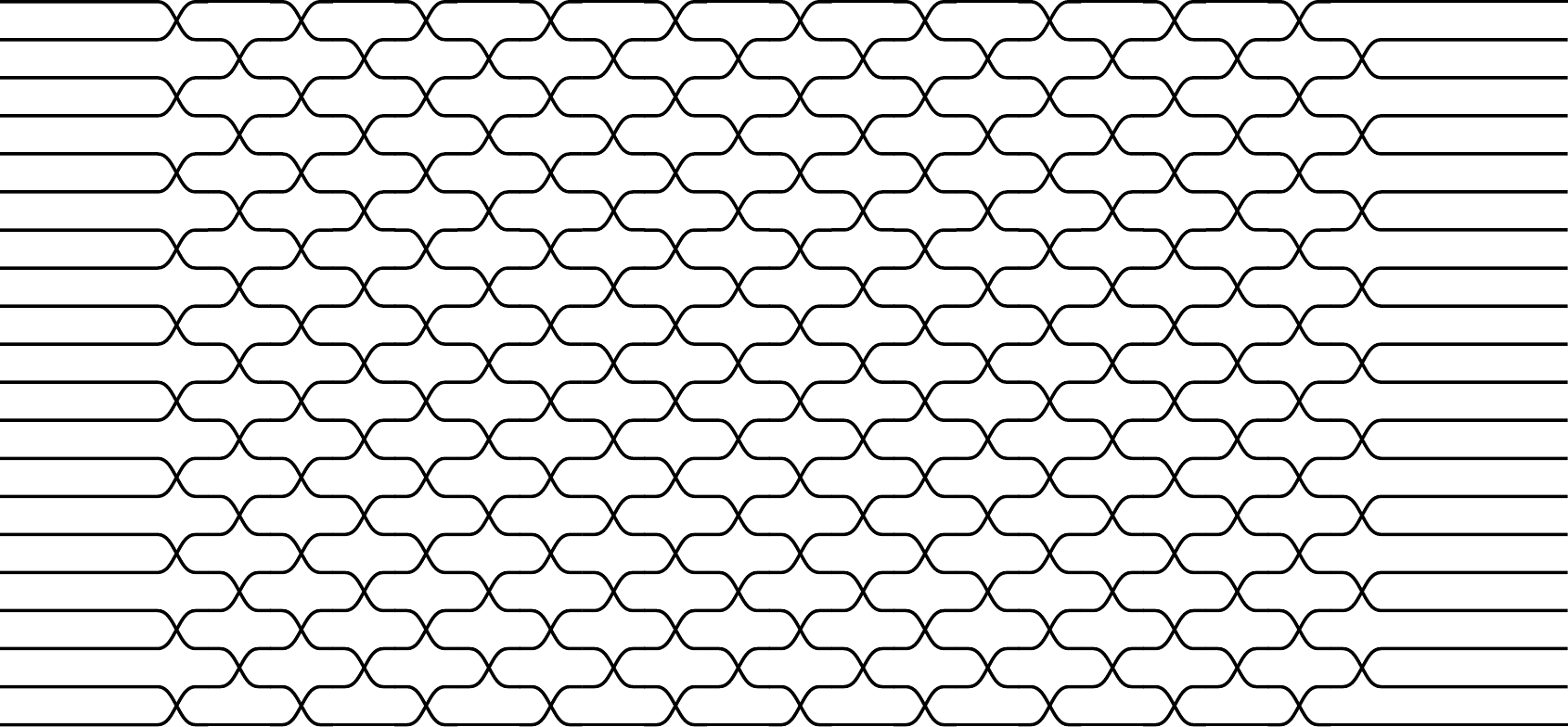}
    \caption{The ``odd-even sort arrangement'' on $20$ pseudolines.}
    \label{fig:bubble20}
\end{figure}

\begin{figure}[h]
    \centering
    \includegraphics[width=0.99\linewidth]{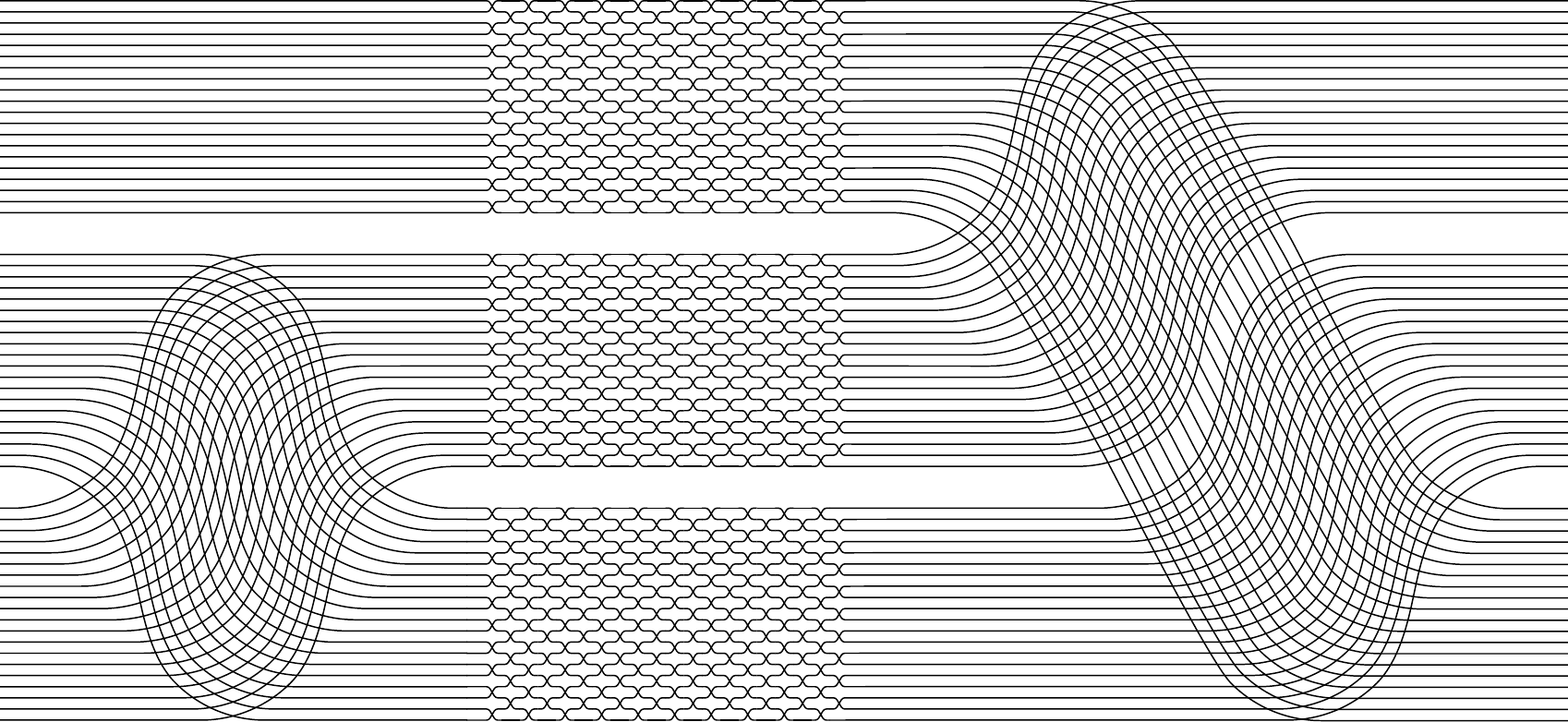}
    \caption{The arrangement $A_3$, obtained by stacking two copies of $A_1$ and adding the missing intersections on the left, then stacking a third copy of $A_1$ on top and adding missing intersections on the right.}
    \label{fig:A3}
\end{figure}
We describe such an arrangement of $n_0 = 2160$ pseudolines. We start with the ``odd-even sort arrangement'' $A_1$ of $20$ pseudolines illustrated in Figure~\ref{fig:bubble20} (who's number of cutpaths can be counted using a formula given by Galambos and Reiner~\cite[Thm.\ 4]{GalambosR08}). For $i \geq 2$ we construct $A_i$ by stacking a copy of $A_1$ above $A_{i-1}$, and adding all missing intersections to the left (resp.\ the right) of them if $i$ is even (resp.\ odd).\footnote{If $A_1$ consisted of a single pseudoline instead of an arrangement of $20$ pseudolines, this procedure would results in so-called ``bubblesort arrangements".} Figure~\ref{fig:A3} illustrates the construction of $A_3$. The arrangement we consider is $A_{108}$, which has $n_0 = 20\cdot 108 = 2160$ pseudolines. It can be checked computationally that this arrangement has more than $2.083^{2160}$ cutpaths\footnote{The exact number is {\npthousandsep{\ }\numprint{23381059977428658244401290752012080837753909224529342117895265969594738964238847418593419609924786300458530435964137137762992929813221336803597097861109775564495459444435804062766740947504211747793597814174740178577284544007168971158885523979884915284797730672401701367033432566980031798653682291712525849971756399918432958213457056998095930747712607108532656455076872297515794022138470093346160678426557238340364776777530322444698136691208450273444733072849224965053685585711998368211884309472894045263460665694690783306626435839163738432420337396543911780171875924849478471799746161814460486526190957001717675329412240320718648002278562080048913942616793104712581227837269762672733591598}.}} (this is simply a matter of counting paths in the dual directed acyclic graph, which can be carried out in a linear number of arithmetic operations in the size of said graph, i.e.\ quadratic in the number of pseudolines). This shows the following.
\begin{theorem}
    For any $n>0$, there exists an arrangement of $n$ pseudolines with $\Omega(2.083^n)$ cutpaths.
\end{theorem}

Note that the starting arrangement of $20$ pseudolines $A_1$ locally maximizes the number of cutpaths in the sense that no triangle flip (an operation by which we move one of the three pseudolines bounding a triangular cell across the intersection of the two others) can increase this number. This is not the case for our final arrangement $A_{108}$, but we have found that the greedy strategy by which we always perform the flip which increases the number of cutpaths the most until reaching a local maximum does not improve our bound significantly. Also, while starting with other sizes for the starting arrangement can lead to better bounds initially (the best previous bound of $2.076^n$ by Ond\v{r}ej Bílka was presumably obtained by considering $n=28$), our choice of $20$ seems to maximize the limit of the bound for our stacking process.

\section{A potential avenue for improvements}

In our upper bound on $B_n$, we have exploited the leading constant of $9.5$ in the tight Zone Theorem. While this constant can not be improved in general, we note that we do not need the full power of this theorem in our proof. Indeed, by appropriately choosing the order in which we encode the pseudolines, it is enough for our purpose to know that in any pseudoline arrangement, at least one of the lines has a small zone complexity (as opposed to all of the lines having at most a certain zone complexity). This would lead to an additional term of order $\Theta(n\lg n)$ in the length (in bits) of the encoding, but this is negligible in front of the leading term of order $\Theta(n^2)$. In this light, we propose the two following conjectures (the second being a weaker consequence of the first).
\begin{conjecture}\label{conj:average_zone}
    Let $\mathcal{L}$ be an arrangement of $n$ pseudolines. Then the average of the zone complexities of the pseudolines in $\mathcal{L}$ is at most $9n+O(1)$.
\end{conjecture}
\begin{conjecture}
    Let $\mathcal{L}$ be an arrangement of $n$ pseudolines. Then there is at least one pseudoline in $\mathcal{L}$ whose zone complexity is at most $9n+O(1)$.
\end{conjecture}

The question of improving the average zone complexity was previously raised by Zerbib~\cite{zerbib2011zone} (without proposing an explicit constant), who showed that a weaker statement implied by such an improvement holds true. Note that we cannot hope for a better constant than $9$ in our conjectures, as one can construct even straight line arrangements meeting this bound (for example by taking all lines supporting the edges of a regular convex $(2k+1)$-gon). If either conjecture holds true, our method would give an upper bound of $B_n < 2^{0.6074n^2}$ for large enough $n$. We note in passing that a confirmation of Conjecture \ref{conj:average_zone} could also for example replace the use of the Zone Theorem in the work of Goaoc and Welzl~\cite{goaoc2023convex}. This would improve the upper bound on the variance of the number of extreme points of a uniformly random labeled $n$-point order type from $3$ to $2+o(1)$.

\bibliography{bibliography}

\end{document}